\documentclass[10pt]{amsart}
\usepackage[british]{babel}
\usepackage{lmodern}
\usepackage[all]{xy}
\usepackage[table,xcdraw]{xcolor}
\usepackage{subfigure}   
\usepackage{mhchem}
\usepackage{tikz,graphicx}
\usepackage{bbm}
\usepackage[justification=justified]{caption}
\usepackage{texdraw,extarrows}
\usepackage[colorlinks,citecolor=blue,urlcolor=blue]{hyperref}
\usepackage{hhline,multirow,mathtools}
\usepackage{amsthm,amsmath,amssymb,amsfonts,mathrsfs,amsfonts,amstext,amsopn}
\usepackage{mathtools}
\makeatletter%
\newcommand\xldownrupharpoon[2][]{%
\ext@arrow 0099{\ldownrupharpoonfill@}{#1}{#2}}
 \def\ldownrupharpoonfill@{%
\arrowfill@\leftharpoondown\relbar\rightharpoonup}
\makeatother

\newcommand{\tS}{\text{S}}

\newcommand{\N}{\mathbb{N}}
\newcommand{\bP}{\mathbf{P}}

\newcommand{\R}{\mathbb{R}}

\newcommand{\cC}{\mathcal{C}}
\newcommand{\cK}{\mathcal{K}}
\newcommand{\cR}{\mathcal{R}}
\newcommand{\cS}{\mathcal{S}}
\newcommand{\sK}{{\sf K}}
\newcommand{\sS}{{\sf S}}
\newcommand{\supp}{{\rm supp\,}}

\makeindex

\makeatletter
\DeclareRobustCommand\widecheck[1]{{\mathpalette\@widecheck{#1}}}
\def\@widecheck#1#2{%
    \setbox\z@\hbox{\m@th$#1#2$}%
    \setbox\tw@\hbox{\m@th$#1%
       \widehat{%
          \vrule\@width\z@\@height\ht\z@
          \vrule\@height\z@\@width\wd\z@}$}%
    \dp\tw@-\ht\z@
    \@tempdima\ht\z@ \advance\@tempdima2\ht\tw@ \divide\@tempdima\thr@@
    \setbox\tw@\hbox{%
       \raise\@tempdima\hbox{\scalebox{1}[-1]{\lower\@tempdima\box
\tw@}}}%
    {\ooalign{\box\tw@ \cr \box\z@}}}
\makeatother

\numberwithin{equation}{section}
\theoremstyle{plain}
\newtheorem{theorem}{Theorem}[section]

\newtheorem{corollary}[theorem]{Corollary}
\newtheorem{proposition}[theorem]{Proposition}

\theoremstyle{definition}
\newtheorem{definition}[theorem]{Definition}
\newtheorem{example}[theorem]{Example}
\theoremstyle{remark}

\allowdisplaybreaks
\begin{document}
\title[Fiber decomposition of DRNs with applications]
{Fiber decomposition of deterministic reaction networks with applications}
\author{Carsten Wiuf}
\author{Chuang Xu}
\address{
Department of Mathematical Sciences\\
University of Copenhagen, Copenhagen\\
2100, Denmark}
\address{
Faculty of Mathematics\\
Technical University of Munich, Munich\\
Garching bei M\"{u}nchen\\
85748, Germany.}
\email{wiuf@math.ku.dk}
\email{xuc@ma.tum.de}

\date{\today}

\noindent

\begin{abstract}
Deterministic reaction networks (RNs) are tools to model diverse biological phenomena characterized by particle systems, when there are abundant number of particles. Examples include but are not limited to biochemistry, molecular biology, genetics, epidemiology, and social sciences.
In this chapter we propose a new type of decomposition of RNs, called \emph{fiber decomposition}. Using this decomposition, we establish lifting of mass-action RNs preserving stationary properties, including \emph{multistationarity} and \emph{absolute concentration robustness}. Such lifting scheme is simple and explicit which imposes little restriction on the reaction networks. We provide examples to illustrate how this lifting can be used to construct RNs preserving certain dynamical properties.
\end{abstract}


\maketitle

\section{Introduction and state of the art}

Reaction networks (RNs) can be regarded as a modelling machinery for many real-world dynamical systems. Examples include networks in epidemiology \cite{B09}, pharmacology \cite{BI09}, ecology \cite{DDK19}, and social sciences \cite{P18}, as well as
 gene regulatory networks \cite{C18},\!  biochemical reaction networks \cite{F19},\! signalling networks \cite{PHPS05},  and metabolic
networks \cite{SZHS12}. An RN is a finite nonempty set of reactions between complexes consisting of species. When there are abundant species and all species are homogeneously well mixed, an RN can be modelled deterministically by ordinary differential equations (ODEs), called \emph{the rate equation}.

\subsection*{Multistationarity of deterministic reaction networks}

A reaction network is \emph{multistationary} if its rate equation admits multiple steady states (subject to the linear subspace the dynamics is confined to) \cite{F19}. Endowed with mass-action kinetics, the rate equation associated with an RN has a polynomial vector field. Hence to determine steady states of an RN amounts to determining zeros of a polynomial, which in general is challenging \cite{F19,FW11}. Several approaches have been proposed to ensure the existence of multiple positive steady states (steady states with positive entries), e.g., based on deficiency theory \cite{HJ72,E98,F19},  injectivity based tests using a Jacobian criterion \cite{CF05,CF06,CF10}, and homotopy and other approaches \cite{CHW08,JS13,CFMW17}.

Lifting of RNs preserving multistationarity is well investigated in the literature \cite{CFRS07,JS13,BP18}. The reason for studying lifting procedures is two-fold. First of all, whether an RN is multistationary might be solved for a smaller/simpler RN and if it is so, then the larger RN of interest is also multistationary by lifting. Secondly, lifting procedures might provide means to construct complex examples of RNs with the same properties as the simpler RN.
The lifting scheme based on the so-called ``atoms of multistationarity'' \cite{JS13} is valid for \emph{fully open continuous-flow stirred-tank reactors} (CFSTRs),\! a network\! in which all chemical species enter the system at constant rates and are removed at rates proportional to their concentrations, that is, there are reactions $0\ce{<=>[][]}\text{S}$ for all species in the RN.

An RN $\cR$ is \emph{nondegenerately} multistationary provided a subnetwork $\widetilde{\cR}\subseteq\cR$ with the same stoichiometric subspace as $\cR$ is so \cite{JS13} (Here ``nondegenerate'' is in the sense of the Jacobian matrix of the vector field). This in particular implies  two networks share the same set of species. The proof  depends on constructing a mapping from every positive steady state (PSS) of the subnetwork $\widetilde{\cR}$ to a nearby point which is a PSS of $\cR$. Based on the construction, the set of PSSs of $\widetilde{\cR}$ is not necessarily a subset of those of $\cR$.

In contrast, our lifting scheme (Theorem~\ref{th-1}) based on the fiber decomposition proposed in this chapter (i) does allow for two RNs to have different sets of species; (ii) under a certain assumption,
 the projection of the set of PSSs of the original reaction network $\cR$ onto the set of species of the subnetwork recovers precisely the set of PSSs of the subnetwork; (iii) the networks are not necessarily CFSTRs.
Nevertheless, the specific construction of lifting does impose certain conditions on the reaction rate constants of the two RNs.

\subsection*{Absolute concentration robustness}

One  interesting property of RNs is  absolute concentration robustness. An RN is \emph{absolute concentration robust} (ACR) if the system has at least one PSS, and all PSSs projected to a given species $\text{S}_i$ are identical (say, $x_i^*$). Such a species $\text{S}_i$ is called an ACR species with $x_i^*$ being the ACR value. Many biological systems have such ACR property, e.g., the EnvZ-OmpR osmoregulatory system, and double-phosphorylation systems of transcriptional regulatory proteins \cite{SF10,CGK20}. This ACR property is closely related to a desirable property in bioengineering, called \emph{robust perfect adaption}, which means that a biological system can adapt after an external stimulus has been applied, and be insensitive to variations in the biochemical parameters of the system \cite{BGK16}.
We remark that ACR is also closely related to sensitivity analysis of parameters for biological systems \cite{SAF09}.

We will show that our lifting procedure also works for ACR. For this we briefly review the literature on ACR. Based on  linear algebra, a simple sufficient condition for a system of deficiency one to be ACR can be  given \cite{SF10} (see Proposition~\ref{prop-SF}). This has recently been extended to
a class of RNs with a weaker condition   than deficiency one \cite{CGK20}. Also lifting of RNs preserving the ACR property  under the former conditions  is discussed therein.
The notion of ACR has likewise been generalized to local ACR and necessary conditions for local ACR have been given \cite{PF20}. An RN is \emph{local ACR} with local ACR species $\text{S}_i$ if the projection of the set of PSSs onto the $i$-th coordinate is nonempty and finite.
We mention that  results for a stochastic analogue of the ACR property are sparse \cite{AEJ14,ACK17,AC19,EK20}.

\section{Notation}

Let $\R$, $\R_{\ge}$, and $\mathbb{Q}$ be the set of real numbers, nonnegative real numbers, and rational numbers, respectively.
Given a finite index set $J\subseteq\N$, for any $x=(x_j)_{j\in J}\in\R^J$, denote $\supp x:=\{j\in J\colon x_j\neq0\}$.

\section{Reaction networks}

In this section, we introduce reaction networks as well as elementary propositions, as prerequisites of the fiber decomposition of reaction networks.

A reaction network (RN) is composed of a triple $(\cS,\cC,\cR)$ of three non-empty finite sets:
\begin{enumerate}
\item[(i)] $\cS=\{\tS_i\}_{1\le i\le d}$ is a set of symbols, termed \emph{species};
\item[(ii)] $\cC\subseteq\N_0^{\cS}$ is a set of linear combinations of species, termed \emph{complexes}, and
\item[(iii)] $\cR\subseteq \cC\times\cC$ is a set of \emph{reactions}. A reaction $(y,y')$ is denoted $y\ce{->[]}y'$. The complex $y$ is called the \emph{reactant} and $y'$ the \emph{product}.
\end{enumerate}

For convention, we assume every species is in some complex and every complex is in some reaction. Hence we also identify an RN with $\cR$ since $\cS$ and $\cC$ can be deduced from $\cR$. We emphasize that $\cR$ is a set without multiplicity, and hence does not contain multiple identical reactions.

A reaction is \emph{degenerate} if its reactant coincides with its product; otherwise it is non-degenerate.
An RN is \emph{degenerate} if it contains degenerate reactions; otherwise, it is \emph{non-degenerate}. Given an RN $\cR$, let $\cR_*\subseteq\cR$ be the subnetwork (with $\cS_*$ and $\cC_*$ being its sets of species and complexes) only consisting of non-degenerate reactions. The concept of an RN herein is more general than the standard one in the literature of chemical reaction network theory (CRNT) \cite{F19} simply because we allow for degenerate reactions.
Other  definitions of RNs, similar to our definition, have also been explored in the literature \cite{DFW17,DWF17}.

The pair $(\cC,\cR)$ forms a (possibly non-simple) digraph referred to as the \emph{reaction graph}. Hence the reaction graph is non-simple and contains a self-loop if and only if there exists a degenerate reaction in the RN. We adopt the convention that every node is strongly connected to itself. Any weakly connected component is called a \emph{linkage class}. All nodes in a strongly connected component (called a \emph{strong linkage class})  are \emph{terminal} if there are no edges from any node in this component to a node in any other strongly connected component. Any node in a terminal strongly connected component is a \emph{terminal complex}; otherwise it is a \emph{non-terminal complex}. Let $\ell_{\cR}$ be the number of linkage classes of $\cR$.

Let $\Omega:=\{y'-y\colon y\ce{->[]}y'\in\cR\}$ be the set of \emph{reaction vectors} and $\Omega_*:=\Omega\setminus\{0\}=\{y'-y\colon y\ce{->[]}y'\in\cR_*\}$. Let $\sS:=\{\sum_{\omega\in\Omega}c_{\omega}\omega\colon c_{\omega}\in\R\}$ be the span of $\Omega$ over $\R$, termed the \emph{stoichiometric subspace} of $\cR$. Note that $\sS=\sS_*:=\{\sum_{\omega\in\Omega_*}c_{\omega}\omega\colon c_{\omega}\in\R\}$ and $\dim\sS\le d$. Define the \emph{deficiency} of a reaction network $\cR$:
\[\delta=\#\cC_*-\ell_{\cR_*}-\dim\sS.\]
Then $\cR$ and $\cR_*$ share the same deficiency. For every $c\in\R^{\cS}_{\ge0}$, let $\sS_c\:=(\sS+c)\cap\R^{\cS}_{\ge0}$ be the \emph{stoichiometric compatibility class} through $c$.

Given two RNs $\cR^1$ and $\cR^2$.
Let $y\ce{->[]}y'\in\cR^1$ and $\widetilde{y}\ce{->[]}\widetilde{y}'\in\cR^2$. We say $y\ce{->[]}y'$ is representable by $\widetilde{y}\ce{->[]}\widetilde{y}'$ and denoted $\{y\ce{->[]}y'\}\preceq \{\widetilde{y}\ce{->[]}\widetilde{y}'\}$ if
\begin{itemize}
\item[i)] $y'-y$ is a multiple of $\widetilde{y}'-\widetilde{y}$, that is, there exists $r\in\R$ such that $y'-y=r(\widetilde{y}'-\widetilde{y})$,
\item[ii)]  $\supp\widetilde{y}\subseteq\supp y$.
\end{itemize}
A subnetwork $\mathcal{A}\subseteq\cR^1$ is representable by $\widetilde{y}\ce{->[]}\widetilde{y}'$, denoted $\mathcal{A}\preceq\widetilde{y}\ce{->[]}\widetilde{y}'$ if every reaction in $\mathcal{A}$ is so. $\cR^1$ is representable by $\cR^2$, denoted by $\cR^1\preceq\cR^2$ if every reaction in $\cR^1$ is representable by one reaction in $\cR^2$.

A\ \, \emph{deterministic\ reaction\ network}\  \, is\ \, a\ \, pair\ \, consisting\ \, of\ \, an\ \, RN\ \, and\ \, a\ \, \emph{kinetics}\ \ $\mathcal{K}=(\lambda_{y\ce{->[]}y'})_{y\ce{->[]}y'\in\cR}$, where $\lambda_{y\ce{->[]}y'}\colon \R^{\cS}_{\ge0}\to\R_{\ge0}$ is the \emph{rate function} of $y\ce{->[]}y'$, expressing the propensity of the reaction to occur. A special kinetics is   \emph{mass-action kinetics},
\begin{equation}\label{mass-action}
\lambda_{y\to y'}(x)=\kappa_{y\ce{->[]}y'}x^y\colon=\kappa_{y\ce{->[]}y'}\prod_{i=1}^dx_i^{y_i},
\end{equation}
where $\kappa_{y\ce{->[]}y'}$ is referred to as the \emph{reaction rate constant}.  Hence, $\lambda_{y\ce{->[]}y'}(x)>0$ if and only if $\supp x\supseteq \supp y$ for $x\in\R^{\cS}$.

For the ease of exposition rather than for generality, we assume throughout that all RNs are endowed with  mass-action kinetics, and hence we also use $\cR$ to refer to the RN with   mass-action kinetics.

The \emph{rate equation} for  a deterministic RN $\cR$ as well as for the corresponding $\cR^*$, characterizing the change in species concentrations over time is then given by the ODE system
\begin{equation}\label{Eq-1}
  \dot{x}=\sum_{y\ce{->[]}y'\in\cR^*}\lambda_{y\ce{->[]}y'}(x)(y'-y).
\end{equation}
Hence for every $c\in\R_0^{\cS}$, $\sS_c$ is an invariant subspace under the flow generated by \eqref{Eq-1}.

\section{Fiber decomposition of RNs}

In this section, we define a fiber decomposition of an RN and use it to construct an explicite lifting scheme from one RN (\emph{reference reaction network}) to another ``larger'' RN (with more species, complexes, and/or reactions) while preserving various stationary properties, including  multistationarity and ACR propery. We mention that converse to lifting, reduction of RNs can be derived \emph{mutatis mutandis}. Hence the main results can be potentially used to simplify large networks in biochemistry and synthetic biology.

\subsection*{Reference RN and base RN}

Given an RN $(\cS,\cC,\cR)$, let $\cS^1\sqcup\cS^2=\cS$ be a partition of $\cS$ into two disjoint sets $\cS^1$ and $\cS^2$. We refer to $\cS^1$ as the \emph{reference} subset of species. Let $\bP_i$ is the natural projection from $\R^{\cS}$ onto $\R^{\cS^i}$ for $i=1,2$. Hence for $i=1,2$, for  $y\in\cC$, $\bP_iy=0$ if $\supp y\cap\cS^i=\varnothing$, and for  $y\to y'\in\cR$, $\bP_i(y\to y')=\bP_iy\to\bP_iy'$ defines a reaction confined to the species set $\cS^i$. Furthermore, let $\bP_i\cR=\{\bP_iy\to\bP_iy'\colon y\to y'\in\cR\}$ (without multiplicity) and $\bP_i\cC=\{\bP_iy\colon y\in\cC\}$.

Hence $(\cS^1,\cC^{\sf b},\cR^{\sf b})$ with $\cC^{\sf b}=\bP_1\cC$ and $\cR^{\sf b}=\bP_1\cR$ forms a new RN, called the \emph{base reaction network} (BRN) of $\cR$, denoted $\cR^{\sf b}$.
Let $\cR^{\sf b}_*\subseteq\cR^{\sf b}$ be the subset of non-degenerate reactions.

In addition to $\cR$, we consider another mass-action RN, termed the \emph{reference} RN, given as $\cR^{\circ}=\{y_k\to y_k'\}_{k\in\sK^{\circ}}$, where  $K_\circ$ is an index set.
Let $\cR^{\circ}_*\subseteq\cR^{\circ}$ be the RN consisting of the non-degenerate reactions of $\cR^{\circ}$ and  $\sK^{\circ}_*\subseteq K^\circ$  its index set.

Now we are ready to come up with a decomposition of $\cR$ w.r.t. the reference RN $\cR^{\circ}$. Assume

\medskip
\noindent{($\mathbf{A1}$)}  $\cR^{\circ}_*\neq\varnothing$.

\medskip
\noindent ($\mathbf{A2}$)  $\cR=\sqcup_{k\in\cK^{\circ}}\cR_{k}$ is a partition in disjoint RNs, such that  for  $k\in\sK^{\circ}_*$, $\bP_1\cR_{k}\preceq y_k\to y_k'$, and  for  $k\in\sK^{\circ}\setminus\sK^{\circ}_*$, $\bP_1\cR_{k}\subseteq\cR^{\sf b}$ consists of degenerate reactions.

\medskip

For $y\to y'\in\cR_{k}$, $k\in\sK^{\circ}_*$, let \begin{equation}\label{Eq-4}
a^{y\to y'}_{k} (y_k'-y_k)=\bP_1(y'-y),\quad a^{y\to y'}_{k}\in\mathbb{Q}.
\end{equation}
Assumption ($\mathbf{A2}$) implies $\cR^{\sf b}\preceq\cR^{\circ}$.
By definition of representability, a reaction in $\cR^{\sf b}_*$ can only be representable by a reaction in $\cR^{\circ}_*$. Moreover, if $\cR^{\circ}$ is non-degenerate, then so are $\cR^{\sf b}$ and $\cR$.

Nevertheless, a BRN of a non-degenerate RN can be degenerate.

\begin{example}
  Consider the non-degenerate RN
   \[\tS_1+\tS_2\ce{->[]}\tS_1.\]
    Let $\cS^1=\{\tS_1\}$. Then its BRN $\{\tS_1\to \tS_1\}$ is an RN consisting of a degenerate reaction. Hence $\cR^{\circ}$ must be degenerate.
\end{example}

A reaction in $\cR^{\sf b}$ can be representable by more than one reaction in $\cR^{\circ}$.
\begin{example}\label{ex-2}
Consider the RN $\cR$:
\[\tS_1+\tS_2+\tS_3\ce{->[]}2\tS_1+\tS_2+\tS_3,\quad 2\tS_1+\tS_2+2\tS_3\ce{->[]}3\tS_1+\tS_2+3\tS_3.\]
Let $\cR^{\circ}$: \[\tS_1+\tS_2\ce{->[]}2\tS_1+\tS_2,\quad 2\tS_1+\tS_2\ce{->[]}3\tS_1+\tS_2.\]
Hence $\cR^{\sf b}=\cR^{\circ}$. Since either reaction in $\cR^{\circ}$ is representable by the other, then either reaction in $\cR^{\sf b}$ is representable by either reaction in $\cR^{\circ}$.
\end{example}

We further emphasize that such a decomposition given in ($\mathbf{A2}$) may not be unique.
\begin{example}
Revisit Example~\ref{ex-2}. Then
 \begin{align*}
  \cR=&\cR_{1}\sqcup\cR_{2},
  \end{align*}
  gives two decompositions with either
  $$\cR_{1}=\{\tS_1+\tS_2+\tS_3\to 2\tS_1+\tS_2+\tS_3\},\ \cR_{2}=\{2\tS_1+\tS_2+2\tS_3\to 3\tS_1+\tS_2+3\tS_3\}$$
  or
  $$ \cR_{1}=\{2\tS_1+\tS_2+2\tS_3\to 3\tS_1+\tS_2+3\tS_3\},\ \cR_{2}=\{\tS_1+\tS_2+\tS_3\to 2\tS_1+\tS_2+\tS_3\},$$
   where $y_1\to y_1'=\tS_1+\tS_2\to 2\tS_1+\tS_2$ and $y_2\to y_2'=2\tS_1+\tS_2\to 3\tS_1+\tS_2$.
\end{example}

We remark that decompositions of reaction networks are proposed in other contexts \cite{HH10,GHMS20,H19}, e.g., for the pursuit of explicit formulas of stationary distributions \cite{H19}.

\subsection*{Fiber decomposition}

  For a reaction $y\to y\in\cR$, we write $\bP_1(y\to y')\oplus \bP_2(y\to y')$ for the direct sum decomposition  $y\to y'=\bP_1y\oplus\bP_2y\to\bP_1y'\oplus \bP_2y'$.

With a decomposition as in ($\mathbf{A2}$) specified, define the associated \emph{fiber  reaction network} (FRN)  $(\cS^2, \cC_{k}, \cR^{{\sf f}}_{k})$ on $\cS^2$
at every reaction $y_k\to y'_k\in\cR^{\circ}$ as
$$\cC_{k}:=\underset{y\to y'\in\cR_{k}}{\cup}\{\bP_2y,\bP_2y'\},\quad \cR^{{\sf f}}_{k}=\bP_2\cR_{k}.$$
Let $\cR^{{\sf f}}_{k,*}\subseteq\cR^{{\sf f}}_{k}$ be the subset of non-degenerate reactions.

Let $\cR^{\sf b}_{k}=\bP_1\cR_{k}$. Hence $\cR^{\sf b}=\cup_{k\in\sK^{\circ}}\cR^{\sf b}_{k}$.
Recall by the definition of the decomposition in ($\mathbf{A2}$), for every $k\in\sK^{\circ}_*$, $\cR^{\sf b}_{k}\preceq y_\to y_k'$.
Note that for two different reactions  $y_k\to y_k', y_{k'}\to y_{k'}'\in \cR^{\circ}$, $\cR^{\sf b}_{k}$ and $\cR^{\sf b}_{k'}$ may have a non-empty intersection or even coincide. Similarly,
$\cR^{\sf f}_{k}$ and $\cR^{\sf f}_{k'}$ may also have a non-empty intersection or  coincide.
Moreover, $\cR$ is non-degenerate if either (i) all FRNs are so or (ii) $\cR^{\sf b}$ is so.

Finally we remark that depending on the choice of $\cS^1$, an RN can have different fiber decompositions in terms of the BRNs together with the FRNs.

\begin{example}
  Consider the RN $\cR$:
  \[\tS_1+\tS_2\ce{->[]}\tS_2+\tS_3\ce{->[]}\tS_1+\tS_3\ce{->[]}\tS_1+\tS_2.\]

(i) Let $\cS^1=\{\tS_1\}$ and $\cR^{\circ}=\{\tS_1\to 0, 0\to \tS_1, \tS_1\to \tS_1\}$. Label  the three reactions in $\cR^{\circ}$ by 1-3 in the given order. Then $\cR^{\sf b}=\cR^{\circ}$ is degenerate. There exists a unique decomposition (irrespective of the reaction rate constants) satisfying ($\mathbf{A2}$) with $\cR_{1}=\{\tS_1+\tS_2\to\tS_2+\tS_3\}$, $\cR_{2}=\{\tS_2+\tS_3\to\tS_1+\tS_3\}$, and $\cR_{3}=\{\tS_1+\tS_3\to\tS_1+\tS_2\}$. Hence the FRNs $\cR_{1}^{\sf f}=\{\tS_2\to\tS_2+\tS_3\}$, $\cR_{2}^{\sf f}=\{\tS_2+\tS_3\to\tS_3\}$, and $\cR_{3}^{\sf f}=\{\tS_3\to\tS_2\}$ are all non-degenerate.

    (ii)  Let $\cS^1=\{\tS_1,\tS_2\}$ and $\cR^{\circ}=\{\tS_1+\tS_2\to\tS_2, \tS_2\to\tS_1,  \tS_1\to\tS_1+\tS_2\}$. Label  the three reactions in $\cR^{\circ}$ by 1-3 in the given order. Hence $\cR^{\sf b}=\cR^{\circ}$ is non-degenerate. However, with a unique decomposition satisfying ($\mathbf{A2}$), $\cR_{1}^{\sf f}=\{0\to\tS_3\}$, $\cR_{2}^{\sf f}=\{\tS_3\to\tS_3\}$, and $\cR_{3}^{\sf f}=\{\tS_3\to0\}$ are \emph{not} all non-degenerate.
\end{example}

\section{Lifting of reaction networks}

Using the setup in the above two sections, we will construct a \emph{larger} RN from a \emph{smaller} RN (the reference RN), so that the set of \emph{positive steady states} of the larger RN projected onto the species set of the smaller one coincides with the set of PSSs of the latter. Such reference RN plays a role as the core module of the larger RN. Specifically, in terms of a fiber decomposition of an RN, we look for an RN of the set $E$ of PSSs with a prescribed reference RN of the set $E^{\circ}$ of PSSs such that $\varnothing\neq E\subseteq E^{\circ}\oplus\R_{>0}^{\cS^2}$.

Based on the fiber decomposition of an RN, given a reference RN $\cR^{\circ}$, there will be diverse ways to construct an RN $\cR$ with $\cR^{\circ}$ as its prescribed BRN. In the following, we propose several ways to construct $\cR$ preserving the aforementioned stationary property.

Recall that we assume mass-action kinetics. For $y\to y'\in\cR$, let $\kappa_{y\to y'}$ denote the corresponding rate constant, and for  $y_k\to y_k'\in\cR^{\circ}$, let $\kappa_{k}^{\circ}$ denote the corresponding rate constant.

Assume

\medskip
\noindent{($\mathbf{A3}$)} $\sum_{y\to y'\in\cR_{k,*}}a^{y\to y'}_{k}
\frac{\kappa_{y\to y'}}{\kappa_{k}^{\circ}}w^{\bP_1y-y_k}v^{\bP_2y}$ is independent of $k\in\sK^{\circ}_*$.

\medskip

Recall that mass-action kinetics of an RN is determined only by the reactants and the reaction rate constants. Hence there is no restriction on the products of $\cR$ (or equivalently, those of the FRNs). Assumption ($\mathbf{A3}$) guarantees that the $\bP_1$-projection of the set $E$ of PSSs of $\cR$ is a subset of the set $E^{\circ}$ of PSSs of $\cR^{\circ}$, due to \eqref{Eq-4}.

In the light of ($\mathbf{A3}$), assume

\medskip
\noindent{($\mathbf{A4}$)} there exists $w\oplus v\in E^{\circ}\oplus\R^{\cS^2}_{>0}$ such that $$\sum_{k\in\sK^{\circ}}\sum_{y\to y'\in\cR_{k,*}}\kappa_{y\to y'}w^{\bP_1y}v^{\bP_2y}\bP_2(y'-y)=0.$$

\medskip
This assumption guarantees that $\varnothing\neq E\subseteq E^{\circ}\oplus\R^{\cS^2}_{>0}$.
It is readily verified that the following assumption  implies ($\mathbf{A4}$) and ensures $E=E^{\circ}\oplus\R^{\cS^2}_{>0}$.

\medskip
\noindent{($\mathbf{A5}$)} There exists a mapping $\sigma\colon \cS^2\to\cS^1$ such that for  $k\in\sK^{\circ}$ and  $y\to y'\in\cR_{k}$, there exists  $b^{y\to y'}_{k}=(b_{k,i}^{y\to y'})_{i\in\cS^2}\in(\mathbb{Q}\setminus\{0\})^{\cS^2}$ such that for  $i\in\cS^2$,
$$\bP_2(y'_i-y_i)=b_{k,i}^{y\to y'}((y_k')_{\sigma(i)}-(y_k)_{\sigma(i)})$$ and $$\sum_{y\to y'\in\cR_{k,*}}b^{y\to y'}_{k,i}
\frac{\kappa_{y\to y'}}{\kappa_{k}^{\circ}}w^{\bP_1y-y_k}v^{\bP_2y}$$ is independent of $k\in\sK^{\circ}_*.$
\medskip

From ($\mathbf{A5}$) it follows that $\cR_{k,*}=\varnothing$ for all $k\in\sK^{\circ}\setminus\sK^{\circ}_*$. In other words, all FRNs at degenerate reactions in the reference RN consist of degenerate reactions.

\begin{example}
  Consider the mass-action RN $\cR^{\circ}$:
  \[\tS_1\ce{<=>[\kappa_1^{\circ}][\kappa_2^{\circ}]}2\tS_1\] with $\sK^{\circ}=\sK^{\circ}_*=\{1,2\}$ consistent with the indices of the reaction rate constants. Hence ($\mathbf{A1}$) is satisfied.
  Consider its lifting $\cR$:
    \[2\tS_1+3\tS_2\ce{->[\kappa_1^{\circ}]}3\tS_1+2\tS_2,\quad 3\tS_1+3\tS_2\ce{->[\kappa_2^{\circ}/2]}\tS_1+5\tS_2,\]
where the labels $\kappa_1^{\circ}$ and $\kappa_2^{\circ}/2$ over the arrows are the rate constants associated with the reactions, which are ordered by the indices of $\kappa_i^{\circ}$, $i=1,2$. Let $\cR=\cR_1\sqcup\cR_2$ with $\cR_i$ composed of the $k$-th reaction of $\cR$ for $k=1,2$. Hence for $k=1,2$ and $y\to y'\in\cR_k$, $a_k^{y\to y'}=k$ and $\bP_1 y-y_k=1$ and $\bP_2y=3$. Hence $\bP_1\cR_{k}\preceq y_k\to y_k'$ for $k=1,2$, and ($\mathbf{A2}$) is satisfied. Moreover, ($\mathbf{A3}$) is satisfied with
  $$\sum_{y\to y'\in\cR_{k,*}}a^{y\to y'}_{k}
\frac{\kappa_{y\to y'}}{\kappa_{k}^{\circ}}w^{\bP_1y-y_k}v^{\bP_2y}=wv^2,\quad \text{for}\ k=1,2.$$
In addition, $\sigma(2)=1$, and $b_{k,2}^{y\to y'}=-k$. Hence  ($\mathbf{A5}$) is satisfied with
$$\sum_{y\to y'\in\cR_{k,*}}b^{y\to y'}_{k,i}
\frac{\kappa_{y \to y'}}{\kappa_{k}^{\circ}}w^{\bP_1y-y_k}v^{\bP_2y}=-wv^2.$$
It is easy to verify that $E^{\circ}=\{\kappa_2^{\circ}/\kappa_1^{\circ}\}$ and $E=\{\kappa_2^{\circ}/\kappa_1^{\circ}\}\oplus\R_{>0}$.
\end{example}

We remark that one can make  a more general assumption than \noindent{($\mathbf{A3}$)}, similar to \noindent{($\mathbf{A5}$)}, by assuming the independence coordinate-wise. Such an assumption, however, will sacrifice the structure of the reference RN $\cR^{\circ}$ as a core module of $\cR$.

\begin{theorem}\label{th-1}
Given a non-degenerate mass-action RN $\cR$. Let $E$ and $E^\circ$ be the set of PSSs of $\cR$ and $\cR^{\circ}$, respectively. Assume ($\mathbf{A1}$)-($\mathbf{A4}$). Then $E\subseteq E^{\circ}\oplus\R^{\cS^2}_{>0}\neq\varnothing$. In particular, assume ($\mathbf{A5}$) additionally, then $E=E^{\circ}\oplus\R^{\cS^2}_{>0}$.
\end{theorem}
\begin{proof}

Rewrite \eqref{Eq-1} as
\begin{equation}\label{Eq-10}
  \dot{x}=\sum_{k\in\sK^{\circ}}\sum_{y\to y'\in\cR_{k\in\sK}}\kappa_{y\to y'}x^y(y'-y).
\end{equation}

Let $x=w\oplus v$ with $w\in\R^{\cS^1}_{\ge0}$ and $v\in\cR^{S^2}_{\ge0}$. In the light of \eqref{Eq-4} and
\begin{equation}\label{Eq-5}
y_k'-y_k=0,\quad \text{for all}\quad k\in\sK^{\circ}\setminus\sK^{\circ}_*,
\end{equation}
rewrite \eqref{Eq-10} as
\begin{equation}\label{Eq-11}
 \begin{split}
  \dot{w}=&\sum_{k\in\sK^{\circ}_*}\kappa_{k}^{\circ}w^{y_k}(y'_k-y_k)\cdot\sum_{y\to y'\in\cR_{k}}
  \frac{\kappa_{y\to y'}}{\kappa_{k}^{\circ}}a_{k}^{y\to y'}w^{\bP_1y-y_k}v^{\bP_2y}
  \end{split}
  \end{equation}
  \begin{equation}\label{Eq-8}
 \begin{split}
  \dot{v}=&\sum_{k\in\sK^{\circ}}\kappa_{k}^{\circ}w^{y_k}\cdot\sum_{y\to y'\in\cR_{k}}
  \frac{\kappa_{y\to y'}}{\kappa_{k}^{\circ}}w^{\bP_1y-y_k}v^{\bP_2y}(\bP_2y'-\bP_2y).
\end{split}
\end{equation}
From \eqref{Eq-11} it follows that ($\mathbf{A3}$) and ($\mathbf{A4}$) together imply that $E\subseteq E^{\circ}\oplus\R^{\cS^2}_{>0}\neq\varnothing$.  

Now assume ($\mathbf{A5}$) additionally. For each $i\in\cS^2$,
\begin{equation*}
 \begin{split}
  \dot{v}_i=&\sum_{k\in\sK^{\circ}_*}\kappa_{k}^{\circ}w^{y_k}((y_k')_{\sigma(i)}-(y_k)_{\sigma(i)})
  \sum_{y\to y'\in\cR_{k,*}}b^{y\to y'}_{k,i}\frac{\kappa_{y\to y'}}{\kappa^{\circ}_{k}}
  w^{\bP_1y-y_k}v^{\bP_2y}\\
  =& \sum_{y\to y'\in\cR_{k,*}}b^{y\to y'}_{k,i}\frac{\kappa_{y\to y'}}{\kappa^{\circ}_{k}}
  w^{\bP_1y-y_k}v^{\bP_2y}\sum_{k\in\sK^{\circ}_*}\kappa_{k}^{\circ}w^{y_k}((y_k')_{\sigma(i)}-(y_k)_{\sigma(i)}).
  \end{split}
  \end{equation*}
     Therefore for all $x^*\in E^{\circ}\oplus\R^{\cS^2}_{>0}\neq\varnothing$, $x^*$ is a PSS for $\cR$. By (i), $E=E^{\circ}\oplus\R^{\cS^2}_{>0}$.
\end{proof}

We propose  more checkable assumptions than ($\mathbf{A3}$) and ($\mathbf{A5}$).

\medskip
\noindent{($\mathbf{A6}$)} The sets $\cC_1=\{\bP_1y-y_k\colon y\to y'\in\cR_{k}\}$ and $\cC_2=\{\bP_2y\colon y\to y'\in\cR_{k}\}$ do not depend on  $k\in\sK^{\circ}$.

\medskip
By ($\mathbf{A6}$), for all $k\in\sK^{\circ}$, $\cR_{k}$ can be decomposed as:
\[\cR_{k}=\sqcup_{z_1\in\cC_1,z_2\in\cC_2}\cR_{k,z_1,z_2},\]
where $\cR_{k,z_1,z_2}=\{y\to y'\in\cR_{k}\colon \bP_1y=z_1+y_k,\ \bP_2y=z_2\}$.

\medskip
\noindent{($\mathbf{A7}$)} There exists a mapping $\sigma\colon \cS^2\to\cS^1$ such that for $k\in\sK^{\circ}$ and $y\to y'\in\cR_{k}$, there exists  $b_{k}^{y\to y'}=(b_{k,i}^{y\to y'})_{i\in\cS^2}\in(\mathbb{Q}\setminus\{0\})^{\cS^2}$ such that for  $i\in\cS^2$,
$$\bP_2(y'-y)_i=b_{k,i}^{y\to y'}((y_k')_{\sigma(i)}-(y_k)_{\sigma(i)}),$$
 both $\sum_{y\to y'\in\cR_{k,z_1,z_2}}a^{y\to y'}_{k}\frac{\kappa_{y\to y'}}{\kappa^{\circ}_{k}}$ and $\sum_{y\to y'\in\cR_{k,z_1,z_2}}b^{y\to y'}_{k,i}\frac{\kappa_{y\to y'}}{\kappa^{\circ}_{k}}$ are non-zero and independent of $k\in\sK^{\circ}_*$.

\begin{theorem}\label{th-2}
Let $\cR$ be a non-degenerate mass-action RN. Let $E$ and $E^{\circ}$ be the set of PSSs of $\cR$ and $\cR^{\circ}$, respectively. Assume ($\mathbf{A1}$)-($\mathbf{A2}$) and ($\mathbf{A6}$)-($\mathbf{A7}$). Then $E=E^{\circ}\oplus\R^{\cS^2}_{>0}$.
\end{theorem}
\begin{proof}
 By ($\mathbf{A6}$) and ($\mathbf{A7}$), we can rewrite \eqref{Eq-11} and \eqref{Eq-8} as
 \begin{equation}
 \begin{split}
   \dot{w}=&\sum_{k\in\sK^{\circ}_*}\kappa_{k}^{\circ}w^{y_k}(y'_k-y_k)
   \sum_{z_1\in\cC_1}w^{z_1}\sum_{z_2\in\cC_2}v^{z_2}\\&\cdot\sum_{y\to y'\in\cR_{k,z_1,z_2}}
   a^{y\to y'}_{k}\frac{\kappa_{y\to y'}}{\kappa^{\circ}_{k}}\\
   \end{split}\end{equation}
   \begin{equation}
   \begin{split}
   \dot{v}_i=&\sum_{k\in\sK^{\circ}_*}\kappa_{k}^{\circ}w^{y_k}((y'_k)_{\sigma(i)}-(y_k)_{\sigma(i)})
   \sum_{z_1\in\cC_1}w^{z_1}\sum_{z_2\in\cC_2}v^{z_2}\\&\cdot\sum_{y\to y'\in\cR_{k,z_1,z_2}}
   b^{y\to y'}_{k,i}\frac{\kappa_{y\to y'}}{\kappa^{\circ}_{k}},\quad i\in\cS^2.
 \end{split}
 \end{equation}
 Since $\sum_{y\to y'\in\cR_{k,z_1,z_2}}a^{y\to y'}_{k}\frac{\kappa_{y\to y'}}{\kappa^{\circ}_{k}}$ and $\sum_{y\to y'\in\cR_{k,z_1,z_2}}b^{y\to y'}_{k,i}\frac{\kappa_{y\to y'}}{\kappa^{\circ}_{k}}$ are both non-zero and independent of $k\in\sK^{\circ}_*$, we have $E=E^{\circ}\oplus\R^{\cS^2}_{>0}$.
\end{proof}

\begin{definition}
An RN is called \emph{multistationary} if there exists a stoichiometric compatibility class with more than one
PSS. Hence potentially an RN may admit multiple PSSs on some stoichiometric compatibility classes while admitting at most one PSS on other stoichiometric compatibility classes, as illustrated by Example~\ref{ex-8} below.
An RN is called \emph{absolute concentration robust} (ACR) if the projection of all positive steady states onto a species $\tS_i$ ($1\le i\le d$) are identical \cite{SF10}.
\end{definition}

\begin{example}\label{ex-8}
  Consider the mass-action RN $\cR$
  \begin{align*}
&\tS_1+\tS_2\ce{->[6]}2\tS_1,\quad \tS_1+2\tS_2\ce{->[4]}3\tS_2,\quad \tS_1+3\tS_2\ce{->[1]}2\tS_1+2\tS_2,\\
&\qquad\qquad\ 2\tS_1+\tS_2\ce{->[2]}3\tS_2,\quad 3\tS_1+\tS_2\ce{->[1]}4\tS_1.
\end{align*}
It is readily verified that the rate equation for $\cR$ is
\begin{align*}
  \dot{w}=&wv(w^2-4w+v^2-4v+6)\\
  \dot{w}=&-wv(w^2-4w+v^2-4v+6).
\end{align*}
Moreover, $\cR$ is conservative with $w(t)+v(t)=w(0)+v(0)$ for all $t>0$. Hence $E=\{(x,y)\in\R^2_{>0}\colon (x-2)^2+(y-2)^2=2\}$ is the set of PSSs. Hence for $c=(c_1,c_2)\in\R^2_{>0}$, $\sS_c$ admits two PSSs (with the left one being an unstable node and the right one stable) if $2<c_1+c_2<6$, one PSS (saddle) if $c_1+c_2=2$ or $c_1+c_2=6$, and no PSS if $0\le c_1+c_2<2$ or $c_1+c_2>6$. See Figure~\ref{figm}.
\end{example}

\begin{figure}%
\centering
\includegraphics[height=2in]{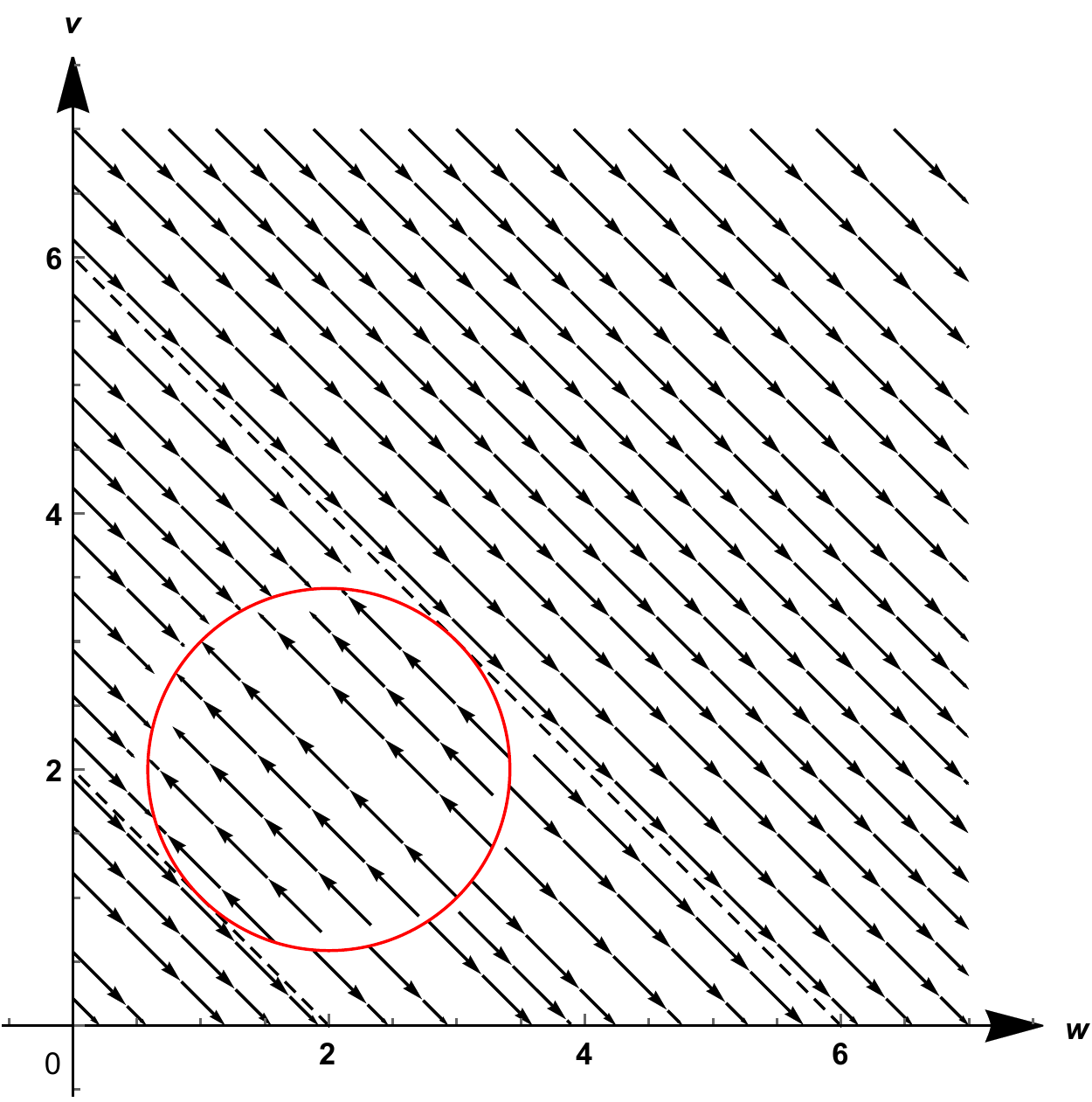}%
\caption{Streamlines for Example~\ref{ex-8}. Red circle: PSSs.}\label{figm}
\end{figure}

The following two corollaries provide a lifting result pertaining to multistationarity and the ACR property.
One can also view the results regarding reductions (of a larger RN), as opposed to lifting (of a smaller RN).

For $c\in\R^{\cS^1}_{\ge0}$, denote $\sS^{\circ}$ and $\sS^{\circ}_c$ the stoichiometric subspace and the stoichiometric compatibility class of the reference RN $\cR^{\circ}$, respectively.

Assume
\medskip

\noindent{($\mathbf{A8}$)} For  $k\in\sK^{\circ}_*$, $\cR^*_{k}\neq\varnothing$.

\begin{corollary}\label{co-1}
Given a non-degenerate mass-action RN $\cR$ with  non-degenerate mass-action reference RN $\cR^{\circ}$. Assume ($\mathbf{A1}$)-($\mathbf{A2}$), ($\mathbf{A8}$), and that $\cR^{\sf b}$ is multistationary. Then $\cR$ is multistationary provided either (i) ($\mathbf{A3}$) and ($\mathbf{A5}$) or (ii) ($\mathbf{A6}$) and ($\mathbf{A7}$). In this case, $\cR$ is called a \emph{multistationarity lifting} of $\cR^{\sf b}$.
\end{corollary}

\begin{proof}
We only prove the conclusions under (i). The other case can be proved analogously. Let $\sS_{\circ}$ be the stoichiometric subspace of $\cR^{\circ}$. Then$$\sS^{\circ}=\left\{\sum_{k\in\sK^{\circ}_*}c_{k}(y_k'-y_k)\colon c_{k}\in\R\right\}.$$
By ($\mathbf{A1}$), ($\mathbf{A2}$) and ($\mathbf{A8}$),  we have $\bP_1\sS=\sS^{\circ}$. For every $c\in\R^{\cS^1}_{\ge0}$, let $E^{\circ}_c$ be the set of PSSs on the stoichiometric compatibility class $\sS_c$. Assume $\cR^{\circ}$ is multistationary. Then $\#E^{\circ}_c>1$ for some $c\in\R^{\cS^1}_{\ge0}$. By Theorem~\ref{th-2}, we have $E=E^{\circ}\oplus\R^{\sS^2}_{>0}$. Choose a $\widetilde{c}\in\R^{\cS}_{\ge0}$ with $\bP_1\widetilde{c}=c$. Hence $E_{\widetilde{c}}=E\cap\sS_{\widetilde{c}}$. Since $\bP_1E_{\widetilde{c}}=E^{\circ}_{\bP_1\widetilde{c}}$, we have $E_{\widetilde{c}}=E^{\circ}_{\bP_1\widetilde{c}}\oplus(\bP_2\sS_{\widetilde{c}}\cap\R^{\sS^2}_{>0})$. By ($\mathbf{A3}$), ($\mathbf{A5}$) and ($\mathbf{A8}$), we have $\bP_2\sS_{\widetilde{c}}\cap\R^{\sS^2}_{>0}\neq\varnothing$. Hence $\#E_{\widetilde{c}}\ge\#E^{\circ}_{\bP_1\widetilde{c}}>1$, i.e., $\cR$ is also multistationary.
\end{proof}

\begin{example}
Consider $\cR$
\[\tS_1+\tS_2\ce{->[\kappa_1]}2\tS_2,\quad 2\tS_1+\tS_2\ce{->[\kappa_2]}3\tS_1,\quad 3\tS_1+\tS_2\ce{->[\kappa_3]}2\tS_1+\tS_2.\] Let $\cS^1=\{\tS_1\}$ and the reference RN $\cR^{\circ}$ be
\[\tS_1\ce{->[\kappa_1]}\varnothing,\quad 2\tS_1\ce{<=>[\kappa_2][\kappa_3]}3\tS_1.\]
Assume $\kappa_2^2>4\kappa_1\kappa_3$. Then it is readily verified that $\#E^{\circ}=2$. Moreover, in this case, $a^{y\to y'}_{k}\equiv1$, $b^{y\to y'}_{k}\equiv-1$. All assumptions ($\mathbf{A1}$), ($\mathbf{A2}$), ($\mathbf{A6}$) and ($\mathbf{A7}$) are satisfied. By Corollary~\ref{co-1}, $E=E^{\circ}\oplus\R$, and $\cR$ is a multistationarity lifting of $\cR^1$.
\end{example}

\begin{corollary}\label{co-2}
Given a non-degenerate mass-action RN $\cR$ with  mass-action reference RN $\cR^{\circ}$. Assume ($\mathbf{A1}$)-($\mathbf{A2}$), and $\cR^{\circ}$ is ACR in a species $\text{S}_i$, $i\in\cS^1$. Then $\cR$ is also ACR in a species $\text{S}_i$ with the same ACR value, provided either (i) ($\mathbf{A3}$) and  ($\mathbf{A4}$)  or (ii) ($\mathbf{A3}$) and ($\mathbf{A5}$) or (iii) ($\mathbf{A6}$) and ($\mathbf{A7}$). In any case, $\cR$ is called an \emph{ACR lifting} of $\cR^{\circ}$.
\end{corollary}
\begin{proof}
  The conclusions follow directly from Theorems~\ref{th-1} and \ref{th-2}.
\end{proof}

\begin{example}\label{ex-12}
  Consider the following RN $\cR$ \cite{SF10},
  \[\tS_1+\tS_2\ce{->[\kappa_1]}2\tS_2,\quad \tS_2\ce{->[\kappa_2]}\tS_1.\] Let $\cS^1=\{\tS_1\}$ and $\cR^{\circ}=\{\tS_1\ce{<=>[\kappa_1][\kappa_2]}0\}$ be the mass-action reference RN.
  This reaction network is the simple closed SIS epidemic contact network ($\tS_1$ represents the number of susceptibles and $\tS_2$ that of the infected). Hence ($\mathbf{A1}$) is satisfied with $\cR^{\sf b}=\cR^{\circ}$. Moreover,  ($\mathbf{A2}$) is satisfied with the decomposition $\cR=\cR_{1}\sqcup\cR_{2}$  where $\cR_{1}=\{\tS_2\ce{->[1]}2\tS_2\}$ and $\cR_{2}=\{\tS_2\ce{->[1]}0\}$. It is readily verified that
  ($\mathbf{A6}$) is satisfied with $\cC_1=\{0\}$ and $\cC_2=\{\tS_2\}$, and ($\mathbf{A7}$) is satisfied with $\sigma(2)=1$,  $b_{1,2}^{\tS_1+\tS_2\to 2\tS_2}=b_{2,2}^{\tS_2\to \tS_1}=-1$, and
  $$\sum_{y\to y'\in\cR_{k,z_1,z_2}}a^{y\to y'}_{k}\frac{\kappa_{y\to y'}}{\kappa^{\circ}_{k}}=1,\quad \sum_{y\to y'\in\cR_{k,z_1,z_2}}b^{y\to y'}_{k,i}\frac{\kappa_{y\to y'}}{\kappa^{\circ}_{k}}=-1.$$
  Since $\cR^1$ is ACR in species $\tS_1$ with ACR value $\frac{\kappa_2}{\kappa_1}$, we have by Corollary~\ref{co-2} that $\cR$ is also ACR in $\tS_1$ with  the same ACR value.
\end{example}

A celebrated result  provides a sufficient condition for ACR regardless of the reaction rate constants \cite{SF10}.

\begin{proposition}\label{prop-SF}
  Let $\cR$ be a non-degenerate mass-action RN. Assume $E\neq\varnothing$ and $\cR$ has deficiency one. If there exist a pair of non-terminal complexes which differ only in species $\tS_i$, then $\cR$ is ACR in $\tS_i$.
\end{proposition}

From the construction of the ACR lifting given in Corollary~\ref{co-2}, the defi- ciency of the RN is generally \emph{not} preserved. Hence one can combine Corollary~\ref{co-2} with Proposition~\ref{prop-SF} to generate RNs of \emph{high deficiency} from deficiency one core modules, as illustrated by the two-species toy example below, which provides a different ACR lifting than given in Corollary~\ref{co-2}.

\begin{example}\label{ex-6}
Consider the following mass-action RN $\cR$:
\[\xymatrix{
5\tS_1+6\tS_2&2\tS_1+3\tS_2 \ar[r]^{\kappa_1} &0 \\
         3\tS_1+3\tS_2 \ar[r]^{\kappa_1} \ar[u]^{\kappa_2} \ar[d]^{\kappa_3}& 3\tS_1+4\tS_2 \ar@{}[r]|{\ce{<=>[\hspace{.15cm}\kappa_2\hspace{.15cm}][\hspace{.15cm}\kappa_1\hspace{.15cm}]}} \ar[d]^{\kappa_{3}}& 2\tS_1+4\tS_2\\
          \tS_1 &  4\tS_1+4\tS_2 & 4\tS_1+3\tS_2 \ar[l]_{\kappa_3} \ar[d]^{\kappa_{2}}\\
          & &4\tS_1+2\tS_2}
\]
\noindent(note that same reaction rates are identical to each other).
The reaction graph associated with $\cR$ contains 10 nodes, 2 linkage classes and $\sS=\R^2$. Hence the deficiency of $\cR$ is $10-2-2=6$.

Assume $\kappa_2>\kappa_3$. Let $\cR^{\circ}=\{0\ce{<=>[\kappa_1][\kappa_{2}]}\tS_1\ce{->[\kappa_3]}2\tS_1\}$ be the reference RN. Label the reactions in $\cR^{\circ}$ by the indices of the reaction rate constants. It is easy to show that the rate equation for $\cR^{\circ}$ is
\[\dot{w}=\kappa_1-(\kappa_{2}-\kappa_3)w.\] Hence the set of PSSs of $\cR^{\circ}$ is given by $E^{\circ}=\{\frac{\kappa_2-\kappa_3}{\kappa_1}\}$, and $\cR^{\circ}$ is ACR (one can also deduce this from Proposition~\ref{prop-SF}).
Hence $\cR$ has the decomposition $\cR=\sqcup_{i=1}^3\cR_i$. It is easy to verify that the rate equation for $\cR$ is
\begin{equation}\label{Eq-9}
\begin{cases}
\dot{w}=(v-2)w^2v^3\left(\kappa_1-(\kappa_{2}-\kappa_3)w\right),\\
\dot{v}=(w-3)w^2v^3\left(\kappa_1-(\kappa_{2}-\kappa_3)w\right),
\end{cases}
\end{equation}
which implies that
\[\frac{{\rm d}w}{{\rm d}v}=\frac{v-2}{w-3},\quad \text{if}\ w\neq3.\] This gives the first integral of \eqref{Eq-9}: $$H(w,v)\colon =(w-3)^2-(v-2)^2,$$ and $\cR$ is an integrable system. Moreover, the set of PSSs $E=E^{\circ}\times\R_{>0}\cup\{(3,2)\}$. Hence $\cR$ is ACR in species $\tS_1$ with ACR value $3$ whenever $\frac{\kappa_2-\kappa_3}{\kappa_1}=3$.

Moreover, the equilibrium $(3,2)$ is a saddle, $(3,v^*)$ are a stable node for $v^*>2$ and an unstable node for $v^*<2$. The invariant manifold through $(3,v^*)$ is $\{(w,v)\in\R_{>0}^2\colon (w-3)^2-(v-2)^2=-(v-v^*)^2\}$, which is not a line (1-d hyperplane) but a hyperbola, for every $v^*\neq2$. This is different from the conservative system in Example~\ref{ex-12}. See Figure~\ref{figh}.
\end{example}

\begin{figure}%
\centering
\includegraphics[height=2in]{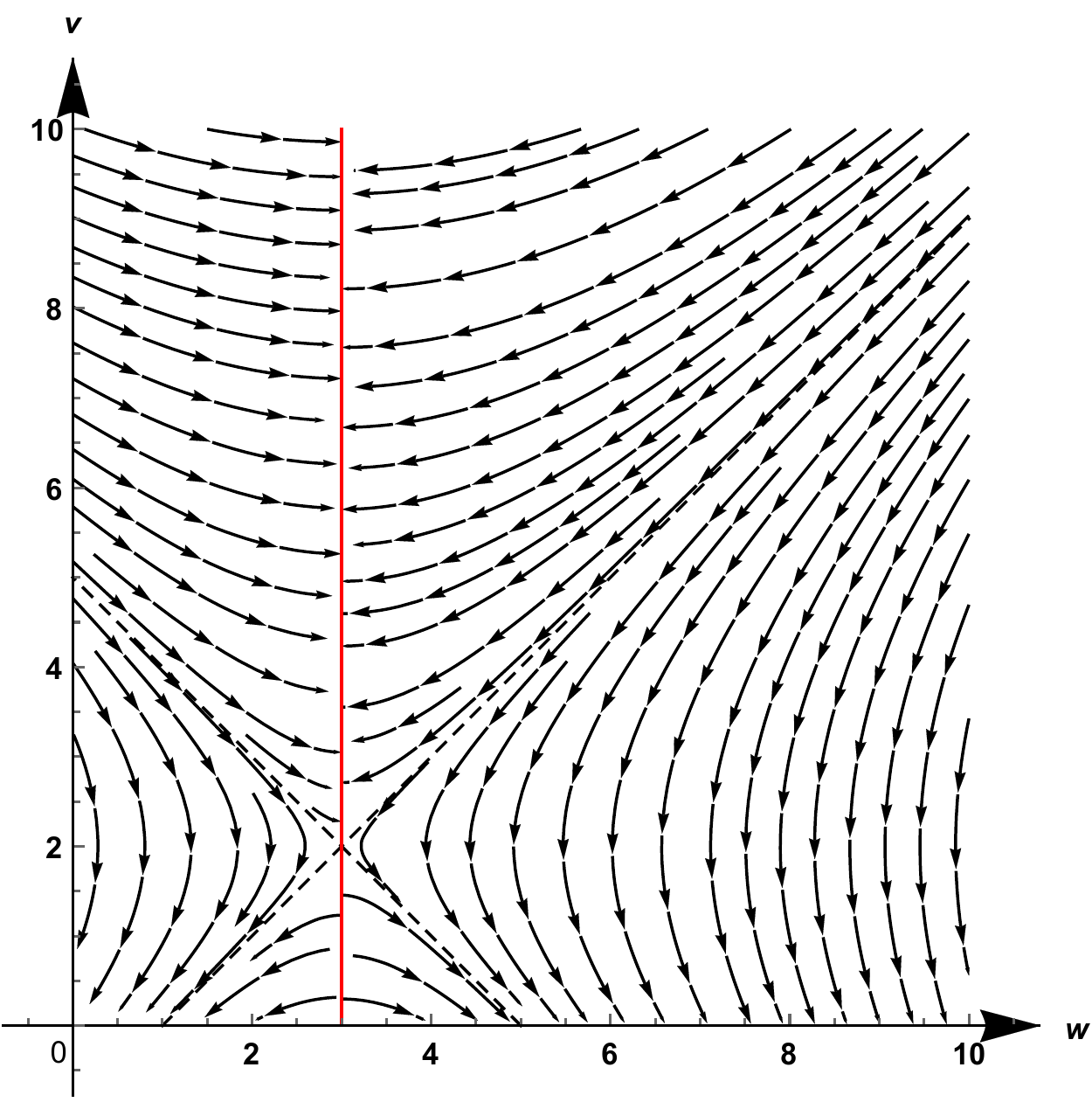}%
\caption{Streamlines for Example~\ref{ex-6} with $\kappa_1=3,\ \kappa_{3}-\kappa_2=1$. Red line: PSSs.}\label{figh}
\end{figure}
Example~\ref{ex-6} illustrates that even certain assumptions (e.g., ($\mathbf{A7}$)) fail, similar lifting still is valid, and can produce stationary dynamics with different geometries. In particular, it demonstrates that the ACR system can have invariant manifolds (characterized in terms of first integrals) which are not hyperplanes (for conservative systems) but hyperbolas, different from what has been observed in the literature (e.g., Example~\ref{ex-12}). This may shed a new light on the study of ACR systems.

Now we present another RN whose invariant manifolds being ellipses.

\begin{example}\label{ex-7}
Consider the following mass-action RN $\cR$:
\[\xymatrix{
2\tS_1+4\tS_2 \ar[r]^{\kappa_1}  &\tS_1+4\tS_2& 4\tS_1+2\tS_2\\
          3\tS_1+4\tS_2 \ar[u]^{\kappa_{3}} \ar[r]^{\kappa_2} &  4\tS_1+4\tS_2 & 4\tS_1+3\tS_2 \ar[l]_{\kappa_3} \ar[u]^{\kappa_{2}}\\
        3\tS_1+3\tS_2 \ar[u]^{\kappa_{1}} \ar[d]^{\kappa_2}   \ar[r]^{\kappa_3}   &    5\tS_1\\
        \tS_1+6\tS_2&2\tS_1+3\tS_2 \ar[r]^{\kappa_1}  &4\tS_1}
\]
(note that some reaction rates are identical). The reaction graph associated with $\cR$ contains 11 nodes, 2 linkage classes and $\sS=\R^2$. Hence the deficiency of $\cR$ is $11-2-2=7$.

Assume $\kappa_2>\kappa_3$. Let $\cR^{\circ}$ defined in Example~\ref{ex-6} be the reference RN. Hence $\cR^{\circ}$ is ACR. Moreover, $\cR$ has the decomposition $\cR=\sqcup_{i=1}^3\cR_i$, and the rate equation for $\cR$ is
\begin{equation}\label{Eq-9}
\begin{cases}
\dot{w}=-(v-2)w^2v^3\left(\kappa_1-(\kappa_{2}-\kappa_3)w\right),\\
\dot{v}=(w-3)w^2v^3\left(\kappa_1-(\kappa_{2}-\kappa_3)w\right),
\end{cases}
\end{equation}
which implies that
\[\frac{{\rm d}w}{{\rm d}v}=-\frac{v-2}{w-3},\quad \text{if}\ w\neq3.\] This gives the first integral of \eqref{Eq-9}: $$H(w,v)\colon =(w-3)^2+(v-2)^2,$$ and $\cR$ is an integrable system. Moreover, the set $E$ of PSSs coincides with that in Example~\ref{ex-6}. Hence $\cR$ is ACR in species $\tS_1$ with ACR value $3$ whenever $\frac{\kappa_2-\kappa_3}{\kappa_1}=3$.

Moreover, the equilibrium $(3,2)$ is a saddle, $(3,v^*)$ are a stable node for $v^*<2$ and an unstable node for $v^*>2$. The invariant manifold through $(3,v^*)$ is $\{(w,v)\in\R_{>0}^2\colon (w-3)^2+(v-2)^2=(v^*-2)^2\}$, which is an ellipse for every $v^*\neq2$. This is also different from the conservative system in Example~\ref{ex-12}. See Figure~\ref{fige}.
\end{example}

\begin{figure}%
\centering
\includegraphics[height=2in]{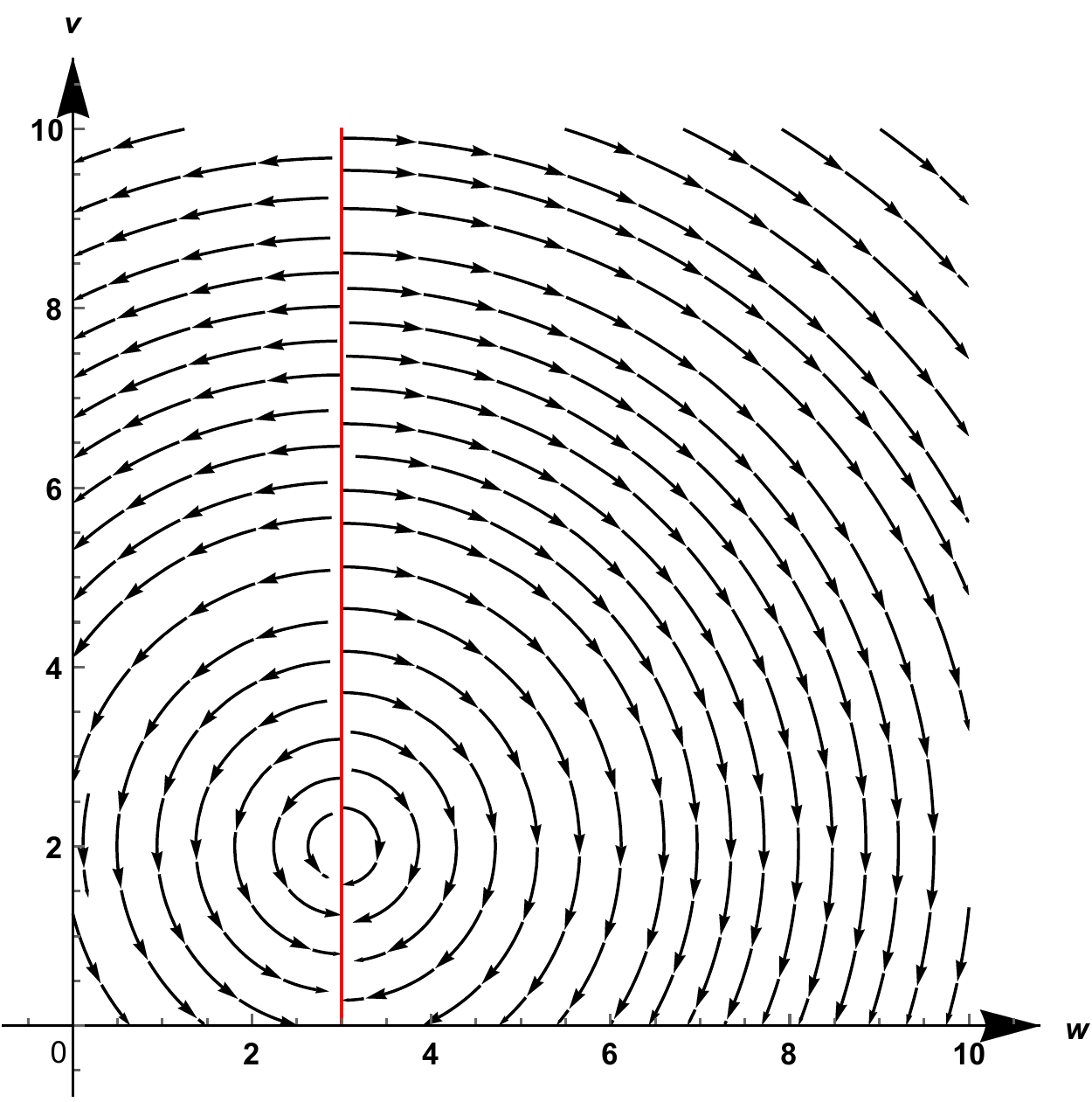}%
\caption{Streamlines for Example~\ref{ex-7} with $\kappa_1=3,\ \kappa_{3}-\kappa_2=1$. Red line: PSSs.}\label{fige}
\end{figure}

\section{Outlooks}

(i) In the main results, we provide conditions ensuring the projection of the set of PSSs of an RN $\cR$ coincides with that of the reference RN $\cR^{\circ}$. Nevertheless, it remains unknown if stability of the PSS of $\cR$ is consistent with that of PSS of $\cR^{\circ}$. This is the case for some known lifting schemes \cite{JS13}.

\noindent(ii) In this chapter, lifting of RNs preserving the ACR property is based on the fiber decomposition of a large RN. Such a decomposition is expected to preserve good properties of RNs. We list a few questions here:
\begin{enumerate}
  \item[(a)] Under what conditions can dynamical/algebraic properties of an RN be deduced barely from its BRN and the FRNs?
  \item[(b)] Are properties consistent for an RN and its BRN and FRNs? For instance, if all FRNs are of deficiency zero, is the original reaction network so? Or, if all FRNs are complex-balanced \cite{HJ72}, are the original reaction network also complex-balanced?
\end{enumerate}

\noindent(iii) The notion of decomposition of RNs has been applied to analysis of  metabolic reaction networks and in  bioinformatics \cite{YSNL07,HH10,KCFS11}. We expect the fiber decomposition may play a role in these regards as well.

\noindent(iv) The fiber decomposition of RNs is readily adapted to \emph{stochastic reaction networks} (by only adjusting the kinetics). In light of the applications of fiber decomposition in the deterministic setting, we believe the analogue for stochastic reaction networks  might also play an important role on similar topics (i.e., lifting stochastic reaction networks preserving for example stationary properties \cite{H19} and \emph{structural classification} \cite{WX20}).
    \begin{example}\label{Ex-10}
      Consider $\cR^{\circ}$:
    \[
    0\ce{<=>[\kappa_1][\kappa_2]}\tS_1\ce{<=>[\kappa_3][\kappa_4]}2\tS_1\ce{->[\kappa_5]}3\tS_1,
    \]
    where $\kappa_i>0$, $\kappa_2>\kappa_3$, $\kappa_4<\kappa_5$ and $(\kappa_3-\kappa_2)^2=4\kappa_1(\kappa_5-\kappa_4)$. The rate equation for $\cR^{\circ}$ is
    \[\dot{w}=(\kappa_5-\kappa_4)(w-w_*)^2,\]where $w_*=\frac{\kappa_2-\kappa_3}{2(\kappa_5-\kappa_4)}$. Hence $\cR^{\circ}$ is ACR.
    Consider the lifting $\cR$ of $\cR^{\circ}$:
    \[\tS_2\ce{->[\kappa_1]}\tS_1,\quad 2\tS_2\ce{<-[\kappa_2]}\tS_1+\tS_2\ce{->[\kappa_3]}2\tS_1,\quad \tS_1+2\tS_2\ce{<-[\kappa_4]}2\tS_1+\tS_2\ce{->[\kappa_5]}3\tS_1.\]
     The rate equation for $\cR$ is
     \begin{align*}
       \dot{w}=&v(\kappa_5-\kappa_4)(w-w_*)^2,\\
       \dot{v}=&-v(\kappa_5-\kappa_4)(w-w_*)^2.
     \end{align*} It is readily verified that $\cR$ is conservative and an ACR-lifting of $\cR^{\circ}$. Nonetheless, by \cite[Theorem~4.6]{WX20}, $\cR$ is positive recurrent on each of its finite compatibility classes while $\cR^{\circ}$ is explosive a.s. on $\N_0$.
    \end{example}

    This example reveals that an ACR lifting may \emph{not} preserve the stochastic dynamics (e.g., explosivity) for the respective stochastic reaction networks.

\section*{Acknowledgements}

CW acknowledges funding from the Novo Nordisk Foundation, Denmark. CX acknowledges the TUM Foundation Fellowship as well as the Alexander von Humboldt Fellowship funded by Alexander von Humboldt Foundation, Germany.

\bibliographystyle{plain}
\bibliography{references}
\end{document}